\documentclass[11pt]{amsart}

\usepackage{fancyhdr}
\usepackage{graphicx,float}
\usepackage{amsmath,amssymb,amsthm,amsfonts,amsxtra,physics}
\pdfoutput=1

\numberwithin{equation}{section} 
\theoremstyle{plain}
\newtheorem{theorem}{Theorem}

\newtheorem{proposition}{Proposition}

\theoremstyle{definition}

\theoremstyle{remark}
\newtheorem{remark}{Remark}

\newcommand{\Ai}{\mathrm{Ai}}

\title[Optimal soft edge scaling]{Optimal soft edge scaling variables for the Gaussian and Laguerre even $\beta$ ensembles}
\author{Peter J. Forrester and Allan K. Trinh}
\address{Department of Mathematics and Statistics, 
	ARC Centre of Excellence for Mathematical \& Statistical Frontiers,
	University of Melbourne, Victoria 3010, Australia}
\email{pjforr@unimelb.edu.au}
\email{a.trinh4@student.unimelb.edu.au}
\date{\today}

\begin{document}
\maketitle
\begin{abstract}
The $\beta$ ensembles are a class of eigenvalue probability densities which generalise the invariant ensembles of classical random matrix theory.	
In the case of the Gaussian and Laguerre weights, the corresponding eigenvalue densities are known in terms of certain $\beta$ dimensional
integrals. We study the large $N$ asymptotics of the density with a soft edge scaling.  In the Laguerre case, this is done with both the parameter $a$ fixed, and with
$a$ proportional to $N$. It is found in all these cases that by appropriately centring the scaled variable, the leading correction term 
to the limiting density is $O(N^{-2/3})$. A known differential-difference recurrence from the theory of Selberg integrals allows for a numerical demonstration of this
effect.
\end{abstract}
	
	\section{Introduction}
	Universal results, whereby scaled large $N$ properties are independent of many local properties, abound in random matrix theory. Due to this, it is possible to talk of the statistical properties of large sized, bulk scaled complex Hermitian matrices, for example, without having to be more specific about the details of the ensemble. Thus the random Hermitian matrices may be formed entry-wise by sampling from a zero mean, finite standard deviation distribution -- this is the Wigner class -- or they may be defined by imposing a unitary invariant measure i.e. depending on the eigenvalues but not the eigenvectors; see e.g. \cite{PS11}.
	
	Bulk scaling refers to choosing the origin in the support of the spectrum but away from the edges, then choosing for the length scale units of the mean spacing in the neighbourhood of this point. This is to be contrasted to the soft edge scaling, for which the origin is chosen to be in the neighborhood of the largest (or possibly smallest) eigenvalue, and the length scale is chosen to be of the order of the spacing between the first and the second eigenvalues. 
	
	Universality results apply in both the bulk and soft edge scaling regimes \cite{Er11}. These two cases show themselves in some celebrated applications of random matrices. In the study of the statistical properties of the large Riemann zeros, the Montgomery-Odlyzko law asserts that upon scaling to have mean spacing unity, these coincide with the statistical properties of the bulk scaled eigenvalues of Hermitian matrices; see e.g. \cite{KS99}. And for growth models on curved interfaces in the KPZ class, fluctuations have been predicted theoretically \cite{PS01} and shown experimentally \cite{Ta11}, to be identical to the fluctuations of the soft edge scaled largest eigenvalue of a random Hermitian matrix. Since many statistical quantities can be computed exactly for bulk and soft edge scaled Hermitian random matrices (see e.g. \cite{Fo10}), there are explicit functional forms available against which to compare data relating to the applications. 
	
	Data on the Montgomery-Odlyzko law comes by way of various data sets of Odlyzko -- originally the high precision evaluation of the $10^{20}$-th Riemann zero and over $70\times 10^6$ neighbours \cite{Od89}, and more recently high precision evaluation of the ($10^{23} + 985,531,550$) zero and over $10^9$ neighbours as announced in \cite{Od01}. The latter sample size is sufficiently large that the difference between the empirical statistical quantity of interest (e.g. the nearest neighbour spacing) as determined from the data, and the random matrix prediction implied by the Montgomery-Odlyzko law, exhibits a distinct functional form \cite{Od01}. In keeping with the proposal of Keating and Snaith \cite{KS00a} that the eigenvalues of random unitary matrices chosen with Haar measure provides the correct statistical description of finite size effects, one can use knowledge of the leading correction to the $N\to\infty$ bulk scaled distribution in this random matrix ensemble to predict the forms seen in the data \cite{BBLM06,FM15,BFM17}.
	
	Soft edge finite size effects can also be analysed \cite{GFF05,Ch06,FFG06,EK06}. While different ensembles exhibit different functional forms for the leading correction, there is growing evidence that the optimal order -- as controlled by the precise choice of the centring variable -- has the leading soft edge correction for complex Hermitian matrices being proportional to $N^{-2/3}$ \cite{FT18}. Although not the subject of the present work, there is also a literature along these lines with respect to hard edge scaling; see \cite{EGP16,Bo16,PS16}.
	
	In the case of real symmetric matrices, the work of Johnstone \cite{Jo08} and collaborators \cite{JM12,Ma12} has shown that the optimal choice of the soft edge centring and scaling variables for the classical ensembles again leads to a correction proportional to $N^{-2/3}$. Here the term classical ensembles refers to the particular class of orthogonally invariant real symmetric random matrices with eigenvalue probability density (PDF) proportional to
	\begin{equation}\label{eq:4.1}
		\prod_{l=1}^N w(x_l)\prod_{1\leq j<k\leq N}\abs{x_k-x_j}
	\end{equation}
	chosen with weight function one of
	\begin{equation}\label{eq:4.2}
		w(x)=
		\begin{cases}
		e^{-x^2/2}, \quad &\text{Gaussian}
		\\
		x^{a/2} e^{-x/2}\chi_{x>0}, \quad &\text{Laguerre}
		\\
		x^a (1-x)^b\chi_{0<x<1}, \quad &\text{Jacobi}
		\end{cases}
	\end{equation}
	(the notation $\chi_A$ denotes the function taking the value $1$ for $A$ true and $0$ otherwise).
	
	The purpose of the present paper is to study the optimal choice of the soft edge scaling for the $\beta$-generalisation of the Gaussian and Laguerre orthogonal ensembles as specified by (\ref{eq:4.1}) and (\ref{eq:4.2}). Now the eigenvalue PDF is proportional to
	\begin{equation}\label{eq:4.3}
	\prod_{l=1}^N w_\beta(x_l)\prod_{1\leq j<k\leq N}\abs{x_k-x_j}^\beta
	\end{equation}
	with weight function
	\begin{equation}\label{eq:4.4}
	w_\beta (x)=
	\begin{cases}
	e^{-\beta x^2/2}, \quad &\text{Gaussian}
	\\
	x^{\beta a/2} e^{-\beta x/2}\chi_{x>0}, \quad &\text{Laguerre}.
	\end{cases}
	\end{equation}
	The matrix ensemble implied by (\ref{eq:4.3}) in the Gaussian case was first introduced by Dyson \cite{Dy62} as an interpolation of the classical Gaussian orthogonal ($\beta=1$), unitary ($\beta=2$) and sympletic ($\beta=4$) ensembles. Later it was shown that in both the Gaussian and Laguerre cases, (\ref{eq:4.3}) for general $\beta>0$ can be realised as the eigenvalue PDF for particular real tridiagonal matrices \cite{DE02}.

	The quantity to be considered is the eigenvalue density
	\begin{equation}\label{eq:5.1}
		\rho_{N,\beta}^\#(x)=w_\beta (x) {N\over Z_{\beta,N}^\#}\int_{\mathbb{R}^{N-1}}\prod_{l=2}^N w_\beta(x_l) \abs{x-x_l}^\beta
		\prod_{2\leq j<k\leq N}\abs{x_k-x_j}^\beta \, dx_2\cdots dx_N,
	\end{equation}
	where $\#=G,(L,a)$ depending on the weight (\ref{eq:4.4}) and
	\begin{equation}\label{eq:5.1a}
		Z_{\beta,N}^\#=\int_{\mathbb{R}^N} \prod_{l=1}^N w_\beta(x_l)\prod_{1\leq j<k\leq N}\abs{x_k-x_j}^\beta \, dx_2\cdots dx_N
	\end{equation}
	is the normalisation. We note that (\ref{eq:5.1}) can be written as
	\begin{equation}\label{eq:5.2}
		\rho_{N,\beta}^\#(x)=w_\beta(x) {N Z_{\beta,N-1}^\# \over Z_{\beta,N}^\#}
		\left\langle
		\prod_{l=2}^N \abs{x-x_l}^\beta
		\right\rangle_{N-1}^\#,
	\end{equation}
	where $\langle \, \cdot \, \rangle_{N-1}^\#$ denotes the average with respect to (\ref{eq:4.3}) in the case of $N-1$ eigenvalues $\{ x_j \}_{j=2}^N$. For $\beta$ even, the average is a polynomial in $\beta$, and it is this feature which gives opportunity to analyse the soft edge scaling limit.
	
	Earlier work of Desrosiers and Forrester \cite{DF06} has, for $\beta$ even, computed the soft edge scaling limit of $\rho_{N,\beta}^\#(x)$ for $\#=G$ and $(L,a)$. This involves the multidimensional integral
	\begin{equation}\label{eq:K}
	K_{n,\beta}(x)=-\frac{1}{(2\pi i)^n}\int_{-i\infty}^{i\infty}dv_1 \cdots \int_{-i\infty}^{i\infty}dv_\beta\, \prod_{j=1}^\beta e^{v_j^3/3-xv_j}\prod_{1\leq k<l\leq n}\abs{v_k-v_j}^{4/\beta},
	\end{equation}
	which extends the integral representation of the Airy function,
	\begin{equation}\label{eq:K1}
		\Ai(x)=\frac{1}{2\pi i}\int_{-i\infty}^{i\infty} e^{v^3/3-xv}\,dv.
	\end{equation}
	It was shown in \cite{DF06} that
	\begin{align}
		\lim_{N\to\infty} \frac{1}{\sqrt{2}N^{1/6}}\rho_{N,\beta}^G\left( \sqrt{2N}+\frac{x}{\sqrt{2}N^{1/6}} \right) &= 
		\lim_{N\to\infty} 2(2N)^{1/3}\rho_{N,\beta}^{(L,a)}\left( 4N+2(2N)^{1/3}x \right) \nonumber
		\\
		&= \frac{1}{2\pi} \left( \frac{4\pi}{\beta} \right)^{\beta/2} \Gamma(1+\beta/2)\prod_{j=1}^{\beta} \frac{\Gamma(1+2/\beta)}{\Gamma(1+2j/\beta)} K_{\beta,\beta}(x) \nonumber
		\\
		&=: \rho_{\infty,\beta,0}(x).
		\label{eq:K2}
	\end{align}
	Here we extend these results to bound the leading correction, by identifying the optimal choice of the centring variable.
	\begin{theorem}\label{T1}
		We have that
		\begin{equation*}
			\frac{1}{\sqrt{2}N^{1/6}}\rho_{N,\beta}^G\left( \sqrt{2N}+\frac{1}{\sqrt{2}N^{1/6}}
			\left(x+ \left( \frac{1}{2}-\frac{1}{\beta} \right) \frac{1}{N^{1/3}}
			\right)
			\right)
		\end{equation*}
		and
		\begin{equation*}
		2(2N)^{1/3}\rho_{N,\beta}^{(L,a)}\left( 4N+2(2N)^{1/3}
		\left(x+\frac{a}{(2N)^{1/3}}
		\right)
		\right)
		\end{equation*}
		are, for large $N$, equal to the RHS of (\ref{eq:K2}) with a correction term $O(N^{-2/3})$.
	\end{theorem}
	An analogous result can also be established for a soft edge scaling in the Laguerre case with $a=\alpha N$.
	\begin{theorem}\label{T2}
		Consider the Laguerre case of (\ref{eq:4.3}) with $a=\alpha N$. Let
		\begin{equation}\label{eq:b}
			b=\left(
			\frac{1}{\sqrt{1+\alpha}}+1
			\right)
			\left(
			\frac{\sqrt{1+\alpha}+1}{2}
			\right)^3.
		\end{equation}
		We have that
		\begin{equation}\label{eq:b1}
			2(bN)^{1/3} \rho_{N,\beta}^{(L,\alpha N)}\left( N(\sqrt{1+\alpha}+1)^2+2(bN)^{1/3}
			\left(x+\frac{\alpha}{\sqrt{1+\alpha}} \left( \frac{1}{2}-\frac{1}{\beta} \right) \frac{1}{2(bN)^{1/3}}
			\right)
			\right)
		\end{equation}
		is, for large $N$, equal to the RHS of (\ref{eq:K2}) with a correction term $O(N^{-2/3})$.
	\end{theorem}

	Moreover, our working below gives the explicit form of the $O(N^{-2/3})$ leading order correction term in each case in terms of generalisations of (\ref{eq:K}). In the Gaussian case this is given by (\ref{eq:14c}) while in the Laguerre case the relevant expressions are (\ref{eq:rLaN}) with $a=\alpha N$ and (\ref{eq:rLa}) with a fixed.
	
	In Section \ref{S2} we revise known results for the leading correction to not only the density, but also the general $k$-point correlation function for the Gaussian and Laguerre unitary ensembles ($\beta=2$ case of (\ref{eq:4.3}) and (\ref{eq:4.4})). Specialised to the density, these results provide a check on our subsequent working. We proceed in Section \ref{S3} to establish the Gaussian case of Theorem \ref{T1}. Section \ref{S4} establishes the Laguerre case of Theorem \ref{T1}, as well as Theorem \ref{T2}. A known differential-difference recurrence for computing the polynomial in (\ref{eq:5.2}) allows for a numerical and graphical demonstration, given in Section \ref{S5}, that the identified scaling variables exhibit convergence at the claimed rate.
	
	\section{The case $\beta=2$}\label{S2}
	A distinguishing feature of the case $\beta=2$ is that the general $k$-point correlation is determinantal,
	\begin{equation}\label{eq:7.1}
		\rho_{(k)}(x_1,...,x_k)=\det[K_N(x_j,x_l)]_{j,l=1}^k
	\end{equation}
	for some correlation kernel $K_N$ independent of $k$. This provides a strategy to establish the leading correction  term to the limiting
	form of $\rho_{(k)}$ 
	under optimal soft edge scaling of this correlation function: thus one should analyse the corresponding asymptotic properties of $K_N$. This was done in the Gaussian and Laguerre cases (the latter with $a$ fixed) by Choup \cite{Ch06} and later extended to the Laguerre unitary ensemble with $a=\alpha N$ in \cite{FT18}, giving an $O(N^{-2/3})$ correction in each case.
	
	To present these results, define
	\begin{equation}\label{eq:7.2}
		s_x=
		\begin{cases}
		\sqrt{2N}+x/\sqrt{2}N^{1/6}, \quad &\text{Gaussian}
		\\
		4N+2a+2(2N)^{1/3}x, \quad &\text{Laguerre}
		\\
		N(\sqrt{1+\alpha}+1)^2+2(bN)^{1/3}x, \quad &\text{Laguerre } a=\alpha N,
		\end{cases}
	\end{equation}
	where in the final formula $b$ is given by (\ref{eq:b}). Also set
	\begin{align}\label{7.3}
	\rho_{\infty,\beta=2,1}(x)&=(\Ai'(x))^2-x(\Ai(x))^2 
	\\
	\rho_{\infty,\beta=2,1}^G(x)&=
	-\frac{1}{20}\left( 3x^2(\Ai(x))^2-2x(\Ai'(x))^2-3\Ai(x)\Ai'(x)\right) \nonumber
	\\
	\rho_{\infty,\beta=2,1}^{(L,a)}(x)&=
	\frac{2^{1/3}}{10}\left( 3x^2(\Ai(x))^2-2x(\Ai'(x))^2+2\Ai(x)\Ai'(x)\right) \nonumber
	\\
	\rho_{\infty,\beta=2,1}^{(L,\alpha N)}(x)&=
	{1\over 2b^{2/3}}\bigg[
	\left( \frac{3}{5} \frac{b}{\sqrt{1+\alpha}} - \frac{3}{20} \frac{\alpha^2}{1+\alpha} \right) x^2\mathrm{Ai}(x)^2
	+ \left( \frac{2}{5} \frac{b}{\sqrt{1+\alpha}} + \frac{1}{40} \frac{\alpha^2}{1+\alpha} \right) \mathrm{Ai}(x)\mathrm{Ai}'(x) \nonumber
	\\
	&\phantom{b^{-2/3} \bigg[\left( \frac{3}{5} \frac{b}{\sqrt{1+\alpha}} - \frac{3}{20} \frac{\alpha^2}{1+\alpha} \right) x^2\mathrm{Ai}(x)^2} 
	+ \left( -\frac{2}{5} \frac{b}{\sqrt{1+\alpha}} + \frac{1}{10} \frac{\alpha^2}{1+\alpha} \right) x\mathrm{Ai}'(x)^2
	\bigg].\nonumber
	\end{align}
	The results of \cite{Ch06,FT18} tell us that for large $N$
	\begin{equation}\label{eq:8.1}
		{\partial s_x\over \partial x} K_N^\#(s_x,s_y) = { \Ai(x)\Ai'(y)-\Ai'(x)\Ai(y)\over x-y } + O(N^{-2/3})
	\end{equation}
	for $\#=G,(L,a),(L,\alpha N)$ and furthermore tell us the precise functional form of the term $O(N^{-2/3})$. To present the latter, we specialise to the case $x=y$, when the LHS of (\ref{eq:8.1}) corresponds to the soft edge scaled density. Using the notation (\ref{7.3}) we have
	\begin{equation}\label{eq:8.2}
		{\partial s_x\over \partial x}\rho_{N,\beta=2}^\#(s_x)
		= \rho_{\infty,\beta=2,0}(x)+\frac{1}{N^{2/3}}\rho_{\infty,\beta=2,1}^\#(x)+\cdots.
	\end{equation}
	
	\section{The Gaussian $\beta$ ensemble}\label{S3}
	Key to analysing (\ref{eq:5.2}) in the Gaussian case for $\beta$ even is the duality formula \cite{BF97a}
	\begin{equation}\label{eq:9}
		\left\langle \prod_{l=1}^N(x-x_l)^n \right\rangle_{\mathrm{ME}_{\beta,N}(e^{-\beta x^2/2})} \,
		= \left\langle \prod_{l=1}^n(x-i x_l)^N \right\rangle_{\mathrm{ME}_{4/\beta,n}(e^{-x^2})}.
	\end{equation}
	Here we have used the notation $\mathrm{ME}_{\beta,N}[w(x)]$ to refer to the PDF corresponding to (\ref{eq:4.3}). Note the exchange of roles $(N,n,\beta)\to (n,N,4/\beta)$ between the two sides of (\ref{eq:9}). We see that the average in (\ref{eq:5.2}) is of the form of the LHS of (\ref{eq:9}) for $\beta$ even, telling us that then
	\begin{equation}\label{eq:9.1}
		\rho_{N,\beta}^G(x)= e^{-\beta x^2/2} \, {N Z_{\beta,N-1}^G\over Z_{\beta,N}^G}
		\left\langle \prod_{l=1}^\beta(x-i x_l)^{N-1} \right\rangle_{\mathrm{ME}_{4/\beta,\beta}(e^{-x^2})}.
	\end{equation}
	The normalisation $Z_{\beta,N}^G$ and that implied by the average in (\ref{eq:9.1}) are known explicitly, see e.g. \cite{BF97a}. Introducing the notation of this latter reference allows (\ref{eq:9.1}) to be written, after minor manipulation,
	\begin{equation}\label{eq:12.0}
		\rho_{N,\beta}^G(\sqrt{2N}x)=\frac{1}{2}{ G_{\beta,N-1}\over G_{\beta,N}\mathcal{N}_\beta }
		(2N)^{\beta N/2+\beta} e^{-\beta Nx^2} R_{N,\beta}(x),
	\end{equation}
	where
	\begin{align*}
		G_{\beta,N}&=g_{\beta,N}\prod_{j=2}^N {\Gamma(1+j\beta/2)\over \Gamma(1+\beta/2)},
		\quad g_{\beta,N}=(2\pi)^{N/2}\beta^{-N(1/2+\beta(N-1)/4)}
		\\
		\mathcal{N}_\beta &= \frac{\pi^{\beta/2}}{2^{\beta-1}} \prod_{j=2}^\beta {\Gamma(1+2j/\beta)\over \Gamma(1+2/\beta)}
	\end{align*}
	and
	\begin{equation}\label{eq:12.1}
		R_{N,\beta}(x)=\int_\mathcal{C} du_1\, e^{Nf(u_1,x)}\cdots\int_\mathcal{C} du_\beta\, e^{Nf(u_\beta,x)}
		\prod_{1\leq j<k\leq \beta}\abs{u_k-u_j}^{4/\beta}
	\end{equation}
	with
	\begin{equation}\label{eq:12.2}
		f(u,x)=-2u^2+\log(iu+x)-\frac{1}{N}\log(iu+x).
	\end{equation}
	The contour $\mathcal{C}$ in (\ref{eq:12.1}) is equal to $\mathbb{R}$ as read off from (\ref{eq:9.1}), however by first ordering $u_1<u_2<\cdots<u_\beta$ the product of differences in the integrand can be interpreted as an analytic function and deformed into the complex plane by Cauchy's theorem (see \cite{DF06}).
	
	Our interest is the large $N$ asymptotics of (\ref{eq:12.0}) after the replacement
	\begin{equation}\label{eq:12.2a}
		x\mapsto 1+\frac{x}{2N^{2/3}},
	\end{equation}
	where the new scaled variable $x$ is $O(1)$. The choice (\ref{eq:12.2a}) is consistent with the first case of
	(\ref{eq:7.2}).  We begin by noting that upon simplification, then use of Stirling's formula, one can check
	\begin{multline}\label{eq:12.3}
		\frac{1}{2}{ G_{\beta,N-1}\over G_{\beta,N}\mathcal{N}_\beta }
		(2N)^{\beta N/2+\beta}
		= \beta^{-\beta/2}2^{-5/2+2\beta +N\beta}\pi^{-1-\beta/2}N^{-1/2+\beta}e^{N\beta/2}\Gamma(1+\beta/2)
		\prod_{j=2}^\beta { \Gamma(1+2/\beta)\over \Gamma(1+2j/\beta)}
		\\
		\times ( 1+O(N^{-1})) .
	\end{multline}
	Also we have
	\begin{equation}\label{eq:12.3a}
		e^{-\beta Nx^2}\bigg\rvert_{x\mapsto 1+x/2N^{2/3}}=e^{-\beta N-\beta xN^{1/3}-\beta x^2/4N^{1/3}}.
	\end{equation}
	In relation to  the exponent in (\ref{eq:12.1}), with $g(u)=-2u^2+\log(iu+1)$ we see from (\ref{eq:12.2}) that
	\begin{multline}\label{eq:12.3b}
		f(u,1+x/2N^{2/3})=g(u)+\frac{1}{N^{1/3}} {1\over 2(iu+1)}-\frac{1}{N} \log(iu+1)-\frac{1}{N^{4/3}} {x^2\over 8(iu+1)^2}
		\\
		-{1\over N^{5/3}} {x\over 2(iu+1)} + O(N^{-1}).
	\end{multline}
	With (\ref{eq:12.3b}) substituted in (\ref{eq:12.1}), the next step is to perform a saddle point analysis of (\ref{eq:12.1}). This has been done to leading order in \cite{DF06} -- here we extend the working therein to the next order but appealing to the latter reference for the formal justification of the methodology. 
	
	According to (\ref{eq:12.3b}) the leading order term in the exponent of (\ref{eq:12.1}) is $g(u)$, which we can check has a double saddle point at $u=u_0:=i/2$. Expanding about this point by introducing the variable $t=u-u_0$, then changing variables $v=2iN^{1/3}t$ -- this is in keeping with the direction of steepest descent; see \cite{DF06} -- gives
	\begin{multline*}
		Nf(i/2+v/2iN^{1/3},1+x/2N^{2/3})=\frac{N}{2}-N\log 2 +N^{1/3} x +\log 2 +v^3/3-v x
		\\
		+\frac{1}{N^{1/3}}\left( -v^4/4+v^2 x-v -x^2/2 \right) 
		+\frac{1}{N^{2/3}}\left( v^5/5-v^3 x+v^2/2 +v x -x \right) + O(N^{-1}).
	\end{multline*}
	Consequently
	\begin{multline}\label{eq:11a}
		\prod_{l=1}^\beta e^{Nf(u_l,1+x/2N^{2/3})}=
		e^{\beta N/2-N\beta \log 2 +\beta xN^{1/3} + \beta\log 2}\, e^{\beta x^2/4N^{1/3}}
		\\
		\times\prod_{j=1}^\beta e^{v_j^3/3-v_j x}
		\left( 1+\frac{1}{N^{1/3}}c_1^G(v;x,\beta)+\frac{1}{N^{2/3}}c_2^G(v;x,\beta)
		+ O(N^{-1})
		\right),
	\end{multline}
	where
	\begin{align}\label{eq:11b}
		c_1^G(v;x,\beta) &= -\frac{3\beta}{4}x^2 + \sum_{j=1}^\beta (-v_j^4/4+v_j^2 x-v_j)
		\\
		c_2^G(v;x,\beta) &= -\beta x + \sum_{j=1}^\beta (v_j^5/5 - v_j^3 x + v_j^2/2 + v_j x^2) + \frac{1}{2}[c_1^G(v;x,\beta)]^2.\nonumber
	\end{align}
	As an extension of (\ref{eq:K}), define
	\begin{equation}\label{eq:K1x}
		\mathbb{K}_{n,\beta}[f](x)=-{1\over (2\pi i)^n}\int_{-i\infty}^{i\infty} dv_1\cdots \int_{-i\infty}^{i\infty}dv_\beta \, f(v_1,...,v_n) \prod_{j=1}^\beta e^{v_j^3/3-xv_j}\prod_{1\leq k<l\leq n} \abs{v_k-v_j}^{4/\beta}.
	\end{equation}
	Substituting (\ref{eq:11a}) in (\ref{eq:12.1}), then substituting the result together with the expansion (\ref{eq:12.3}) in (\ref{eq:12.0}) with $x$ replaced by (\ref{eq:12.2a}), we deduce that
	\begin{multline}\label{eq:K2x}
	{1\over \sqrt{2}N^{1/6}} \rho_{N,\beta}^G \left(\sqrt{2N} + \frac{x}{\sqrt{2}N^{1/6}}\right) 
	= \frac{1}{2\pi}\left(\frac{4\pi}{\beta} \right)^{\beta/2} \Gamma(1+\beta/2)
	\prod_{j=2}^\beta {\Gamma(1+2/\beta) \over \Gamma(1+2j/\beta)}
	\\
	\times\mathbb{K}_{\beta,\beta}
	\left[ 1+ \frac{1}{N^{1/3}}c_1^G(v;x,\beta) + \frac{1}{N^{2/3}}c_2^G(v;x,\beta) +O(N^{-1})
	\right](x).
	\end{multline}
	Our task now is to obtain a simplification of $\mathbb{K}_{\beta,\beta}[c_1^G(v;x,\beta)](x)$.
	\begin{proposition}\label{P1}
		We have
		\begin{equation}\label{eq:K3x}
			\mathbb{K}_{\beta,\beta}[c_1^G(v;x,\beta)](x)=\left( \frac{1}{\beta}-\frac{1}{2}\right) \frac{d}{dx}K_{\beta,\beta}(x).
		\end{equation}
	\end{proposition}
	\begin{proof}
		In (\ref{eq:K1x}), let is regard $f$ as an operator acting on all terms in the integrand to the right.
		
		It follows from the fundamental theorem of calculus that
		\begin{equation*}
			0=\mathbb{K}_{n,\beta}\bigg[ \sum_{j=1}^n {\partial \over \partial v_j} \bigg](x).
		\end{equation*}
		On the other hand, by direct differentiation
		\begin{align*}
			\mathbb{K}_{n,\beta}\bigg[ \sum_{j=1}^n {\partial \over \partial v_j} \bigg](x)
			&= \mathbb{K}_{n,\beta}\bigg[ \sum_{j=1}^n v_j^2 \bigg](x)
			-nx K_{n,\beta}(x) + \frac{4}{\beta}\mathbb{K}_{n,\beta}\bigg[ \sum_{\substack{j,k=1\\j\neq k}}^n {1 \over v_j-v_k} \bigg](x)
			\\
			&=\mathbb{K}_{n,\beta}\bigg[ \sum_{j=1}^n v_j^2 \bigg](x)
			-nx K_{n,\beta}(x),
		\end{align*}
		where the second equality follows from the fact that the integrand implied by the final term in the line above is anti-symmetric upon the interchange of any two co-ordinates and thus vanishes. Hence
		\begin{equation}\label{eq:K4x}
			\mathbb{K}_{n,\beta}\bigg[ \sum_{j=1}^n v_j^2 \bigg](x)
			=nx K_{n,\beta}(x).
		\end{equation}
		It similarly follows from the fundamental theorem of calculus that
		\begin{equation*}
		0=\mathbb{K}_{n,\beta}\bigg[ \sum_{j=1}^n {\partial \over \partial v_j} v_j^2\bigg](x),
		\end{equation*}
		while direct differentiation gives
		\begin{equation*}
			0=2\mathbb{K}_{n,\beta}\bigg[ \sum_{j=1}^n v_j \bigg](x)
			+\mathbb{K}_{n,\beta}\bigg[ \sum_{j=1}^n v_j^4 \bigg](x)
			-x \mathbb{K}_{n,\beta}\bigg[ \sum_{j=1}^n v_j^2 \bigg](x)
			+ \frac{4}{\beta}\mathbb{K}_{n,\beta}\bigg[ \sum_{\substack{j,k=1\\j\neq k}}^n {v_j^2 \over v_j-v_k} \bigg](x).
		\end{equation*}
		Taking the arithmetic mean of the final term with itself after interchanging $v_j\leftrightarrow v_k$ shows
		\begin{equation*}
			\mathbb{K}_{n,\beta}\bigg[ \sum_{\substack{j,k=1\\j\neq k}}^n {v_j^2 \over v_j-v_k} \bigg](x)
			= (n-1)\mathbb{K}_{n,\beta}\bigg[ \sum_{j=1}^n v_j\bigg](x).
		\end{equation*}
		Hence
		\begin{equation}\label{eq:K5x}
			\mathbb{K}_{n,\beta}\bigg[ \sum_{j=1}^n v_j^4 \bigg](x)=
			x\mathbb{K}_{n,\beta}\bigg[ \sum_{j=1}^n v_j^2 \bigg](x)
			-\left( {4\over \beta}(n-1)+2\right)\mathbb{K}_{n,\beta}\bigg[ \sum_{j=1}^n v_j \bigg](x).
		\end{equation}
		
		Recalling the definition of $c_1^G$ in (\ref{eq:11b}), it follows from (\ref{eq:K4x}) and (\ref{eq:K5x}) that
		\begin{equation*}
			\mathbb{K}_{\beta,\beta}[c_1^G(v;x,\beta)](x)=\left( {1\over 2}-{1\over \beta} \right)
			\mathbb{K}_{\beta,\beta}\bigg[ \sum_{j=1}^\beta v_j \bigg](x).
		\end{equation*}
		The result (\ref{eq:K3x}) now follows from the fact that
		\begin{equation}\label{eq:K6x}
			\mathbb{K}_{\beta,\beta}\bigg[ \sum_{j=1}^\beta v_j \bigg](x)=-\frac{d}{dx}K_{\beta,\beta}(x).
		\end{equation}
	\end{proof}
	The results (\ref{eq:K2x}), (\ref{eq:K2}) and (\ref{eq:K3x}) tell us that for $\beta$ even
	\begin{multline}\label{eq:14a}
		\frac{1}{\sqrt{2}N^{1/6}}\rho_{N,\beta}^G\left( \sqrt{2N}+\frac{x}{\sqrt{2}N^{1/6}}
		\right)
		=\rho_{\infty,\beta,0}(x) + \frac{1}{N^{1/3}}\left( {1\over \beta}-{1\over 2} \right)\frac{d}{dx}\rho_{\infty,\beta,0}(x)
		\\
		+ \frac{1}{N^{2/3}}r_{\infty,\beta}^G(x) + O(N^{-1}),
	\end{multline}
	where
	\begin{equation}\label{eq:14b}
		r_{\infty,\beta}^{\#}(x)=
		\frac{1}{2\pi}\left(\frac{4\pi}{\beta} \right)^{\beta/2} \Gamma(1+\beta/2)
		\prod_{j=2}^\beta {\Gamma(1+2/\beta) \over \Gamma(1+2j/\beta)}
		\mathbb{K}_{\beta,\beta}[c_2^{\#}(v;x,\beta)](x).
	\end{equation}
	The first statement in Theorem \ref{T1} now follows from (\ref{eq:14a}) by recentring $x$ according to $x\mapsto x+\left( {1\over 2}-{1\over \beta}\right)$ then expanding the first term according to Taylor's theorem up to second order, and the second term up to first order; this process also gives the explicit $O(N^{-2/3})$ correction term as
	\begin{equation}\label{eq:14c}
		\frac{1}{N^{2/3}}\left(
		-\frac{1}{2}\left( {1\over 2}-{1\over \beta}\right)^2\frac{d^2}{dx^2}\rho_{\infty,\beta,0}(x)+r_{\infty,\beta}^G(x)
		\right).
	\end{equation}
	
	\begin{remark}\label{R1} (A.) In the case $\beta = 2$ we have from (\ref{eq:14c}), (\ref{eq:14b}) and (\ref{eq:K1x}) that the correction term 
	(\ref{eq:14c}) is equal to
	$$
	{1 \over N^{2/3}} {1 \over 2} {1 \over (2 \pi)^2}
	\int_{-i \infty}^{i \infty} dv_1 \, e^{v_1^3/3 - x v_1}  \int_{-i \infty}^{i \infty} dv_2 \, e^{v_2^3/3 - x v_2} \,
	c_2^G(v;x,\beta) | v_2 - v_1|^2,
	$$
	where $c_2^G$ is given by  (\ref{eq:11b}). Expanding $|v_2 - v_1|^2$ and making use of (\ref{eq:K1}) gives agreement with
	$\rho_{\infty,\beta = 2, 1}(x)$ as specified in (\ref{7.3}). \\
	(B.) The factor $( {1 \over \beta} - {1 \over 2})$ in (\ref{eq:K3x}) is familiar as a factor in the $O(1/N)$ corrections to the bulk smoothed
	densities in the Gaussian and Laguerre $\beta$ ensembles; see \cite{Jo98} and the recent works \cite{WF14}, \cite{FRW17}.
	\end{remark}
	
	\section{The Laguerre $\beta$ ensemble}\label{S4}
	Introduce the normalisations
	\begin{align*}
		W_{a,\beta,N}&=w_{a,\beta,N}\prod_{j=1}^N {\Gamma(1+j\beta/2)\Gamma(1+(a+j-1)\beta/2) \over \Gamma(1+\beta/2)},
		\quad
		w_{a,\beta,N}=(2/\beta)^{N(a\beta/2+1+\beta(N-1)/2)}
		\\
		M_n(A,B,C) &= \prod_{j=1}^n {\Gamma(1+A+B-C+jC)\Gamma(1+jC)\over \Gamma(1+A-C+jC)\Gamma(1+B-C+jC)\Gamma(1+C)}.
	\end{align*}
	Use the notation CE${}_{\beta,N}^{(c,d)}$ to denote the ensemble of eigen-angles specified by the PDF proportional to
	$$
	\prod_{l=1}^N e^{\pi i \theta_l(c-d)} | 1 + e^{2 \pi i \theta_l} |^{c+d} \prod_{1 \le j < k \le N}
	|e^{2 \pi i \theta_k} - e^{2 \pi i \theta_j}|^\beta, \qquad -{1 \over 2} < \theta_l < {1 \over 2}.
	$$
	Combining \cite[Prop.~13.2.5 \& Exercises 13.1 (4.4)]{Fo10} gives the duality formula
	\begin{equation}\label{Bv}
	{ 1 \over  W_{a+2n/\beta,\beta,N}}
	\Big \langle \prod_{l=1}^N (x_l - x)^n \Big \rangle_{{\rm ME}_{\beta,N}(x^{a\beta/2} e^{-\beta x/2}) }=
	\Big \langle \prod_{l=1}^n e^{-x e^{2 \pi i \theta}} \Big \rangle_{{\rm CE}_{4/\beta,n}^{(a+2/\beta - 1, N - 1)}}.
	\end{equation}
	
According to \cite{DF06} the multi-dimensional integral implied by the RHS of (\ref{Bv}) can be written in the form (\ref{eq:12.1}), 
but with
\begin{equation}\label{eq:L2}
		f(u,x)=4xu-\log u +\log(1-u)+\frac{1}{N}\left( a-2+\frac{2}{\beta} \right)\log(1-u) + \frac{1}{N}\left( \frac{2}{\beta}-2 \right)\log u,
	\end{equation}
	and $\mathcal C$ a counterclockwise closed contour about $u=1$. Furthermore, as in the Gaussian case, by suitably ordering the contours
	the absolute value signs about $|u_k - u_j|$ can be removed and the contours deformed into the complex plane. Given this,
	after recalling (\ref{eq:5.2}), and using the notation  $R_{N,\beta}$ as is specified in terms of $f(u,x)$ and $\mathcal C$ by (\ref{eq:12.1}),
	we have in keeping with \cite[Eq.~(22)]{DF06} that
	\begin{equation}\label{eq:L1}
		\rho_{N,\beta}^{(L,a)}(4Nx)={N\over (2\pi i)^\beta } {W_{a+2,\beta,N-1}\over W_{a,\beta,N}}
		{ (4Nx)^{a\beta/2}e^{-2\beta Nx}\over M_\beta(a+2/\beta-1,N-1,2/\beta) } R_{N,\beta}(x).
	\end{equation}
		
	\subsection*{The case $a = \alpha N$}
	In keeping with the third case of (\ref{eq:7.2}), we want to expand for large $N$ to the first two orders (\ref{eq:L1}) with
	$a = \alpha N$ and
	\begin{equation}\label{eq:L4}
	x\mapsto \frac{1}{4}(\sqrt{1+\alpha}+1)^2+{b^{1/3}x\over 2N^{2/3}}.
	\end{equation}
	
	We begin by noting
	\begin{multline}\label{eq:L5}
		(4Nx)^{\alpha\beta N/2}e^{-2\beta Nx} \, \bigg\rvert_{x\mapsto \frac{1}{4}(\sqrt{1+\alpha}+1)^2+b^{1/3}x/2N^{2/3}}
		\\
		= N^{\alpha\beta N/2}(\sqrt{1+\alpha}+1)^{\alpha\beta N} \exp(-(bN)^{1/3} {2\beta x \over \sqrt{1+\alpha}+1} -\frac{1}{2}\beta N(\sqrt{1+\alpha}+1)^2 )
		\\
		\times \left(
		1 - \frac{1}{N^{1/3}}\frac{\alpha \beta b^{2/3} x^2}{(\sqrt{1+\alpha}+1)^4}
		+ \frac{1}{N^{2/3}} \frac{\alpha^2 \beta^2 b^{4/3} x^4}{2(\sqrt{1+\alpha}+1)^8}+ O(N^{-1})
		\right)
	\end{multline}
	and
	\begin{multline}\label{eq:L6}
		{W_{\alpha N+2,\beta,N-1}\over W_{\alpha N,\beta,N}} {1\over M_\beta(\alpha N+2/\beta-1,N-1,2/\beta) }
		\\
		= {1\over 2\pi} \left( {4\pi \over \beta} \right)^{\beta/2} N^{-2+\beta- \alpha \beta N/2} e^{\alpha\beta N/2} (1+\alpha)^{(-1+\beta -\beta N - \alpha \beta N)/2} 
		\\
		\times \Gamma(1+\beta/2)
		\prod_{j=2}^\beta {\Gamma(1+2/\beta) \over \Gamma(1+2j/\beta)}
		\Big (
		1 + O(N^{-1}) \Big ).
	\end{multline}
	In (\ref{eq:L6}) use was made use of Stirling's formula and the  multiplication formula for the gamma function
	\begin{equation*}
		\prod_{j=0}^{n-1} \Gamma\left( z+\frac{j}{n} \right) = (2\pi)^{(n-1)/2}n^{1/2-nz}\Gamma(nz).
	\end{equation*}
	
	Next, we turn to a saddle point analysis of the integral in (\ref{eq:L1}). After making the substitution (\ref{eq:L4}) in  (\ref{eq:L2})
	we see that up to terms of order $1/N$, the latter quantity is equal to
	$$
	g(u)=(\sqrt{1+\alpha}+1)^2u-\log u + (1+\alpha)\log(1-u).
	$$
	This possesses a double saddle point at $u_0=1/(\sqrt{1+\alpha}+1)$. 
	We expand about this point in the direction of steepest descent by the
	changing of variables $v=-2(bN)^{1/3}(u-u_0)$. This gives
	\begin{multline}\label{eq:L7}
		Nf\left( u_0 - \frac{v}{2(bN)^{1/3}}, {1 \over 4}(\sqrt{1+\alpha}+1)^2+ {b^{1/3}x\over 2N^{2/3}} \right)
		= 
		\\
		N g(u_0)
		+ (bN)^{1/3} {2x \over \sqrt{1+\alpha}+1} 
		+ (2/\beta-2)(\log u_0+\log(1-u_0)) + v^3/3-v x
		\\
		+ {1\over (bN)^{1/3}} {\alpha \over \sqrt{1+\alpha}} \left( 
		v^4/8+ (1-1/\beta)v
		\right)
		+ {1\over (bN)^{2/3}}\bigg[
		\left( \frac{1}{10}\frac{b}{\sqrt{1+\alpha}} + \frac{3}{80}\frac{\alpha^2}{1+\alpha} \right)v^5
		\\
		+ (1-1/\beta) \left( \frac{b}{\sqrt{1+\alpha}} + \frac{1}{8}\frac{\alpha^2}{1+\alpha} \right)v^2
		\bigg] + O(N^{-1}) .
	\end{multline}
	Combining the results (\ref{eq:L5}), (\ref{eq:L6}), (\ref{eq:L7}) with (\ref{eq:L1}) shows
	\begin{multline}\label{eq:L8}
		2(bN)^{1/3} \rho_{N,\beta}^{(L,\alpha N)} \left( N(\sqrt{1+\alpha}+1)^2 + 2b^{1/3}N^{1/3} x \right) 
		= \frac{1}{2\pi}\left(\frac{4\pi}{\beta} \right)^{\beta/2} \Gamma(1+\beta/2)
		\prod_{j=2}^\beta {\Gamma(1+2/\beta) \over \Gamma(1+2j/\beta)}
		\\
		\times\mathbb{K}_{\beta ,\beta}
		\left[ 1+ 
		\frac{1}{N^{1/3}}c_1^{(L,\alpha N)}(v;x,\beta) 
		+ \frac{1}{N^{2/3}}c_2^{(L,\alpha N)}(v;x,\beta) +O(N^{-1})
		\right],
	\end{multline}
	where
	\begin{multline} \label{eq:L9}
	b^{1/3} c_1^{(L,\alpha N)}(v;x,\beta) 
	= -\frac{\alpha \beta x^2}{8\sqrt{1+\alpha}} + \frac{\alpha}{\sqrt{1+\alpha}} \sum_{j=1}^\beta \Big (v_j^4/8 + (1-1/\beta)v_j \Big ),
	\\
	b^{2/3} c_2^{(L,\alpha N)}(v;x,\beta) 
	= \frac{b}{\sqrt{1+\alpha}}\sum_{j=1}^\beta \Big ( v_j^5/10 + (1-1/\beta)v_j^2 \Big )
	+ \frac{\alpha^2}{1+\alpha}\sum_{j=1}^\beta \Big ( 3v_j^5/80 + (1-1/\beta)v_j^2/8 \Big )
	\\
	+ \frac{1}{2}b^{2/3} [c_{1}^{(L,\alpha N)}(v;x,\beta) ]^2.
	\end{multline}
	
	Analogous to the result of Proposition \ref{P1}, the $O(N^{-1/3})$ term in (\ref{eq:L8}) can be identified as being proportional to the
	derivative of the leading term.
	
	\begin{proposition}\label{P2}
		We have
		\begin{equation}
			\mathbb{K}_{\beta,\beta}[c_1^{(L,\alpha N)}(v;x,\beta)](x)
			={1\over 2b^{1/3}} {\alpha \over \sqrt{1+\alpha}} \left( {1\over \beta}-{1\over 2} \right)
			\frac{d}{dx}K_{\beta,\beta}(x). \label{eq:P2b}
		\end{equation}
	\end{proposition}
	\begin{proof}
	This follows directly from (\ref{eq:L9}), (\ref{eq:K4x}), (\ref{eq:K5x}) and (\ref{eq:K6x}).
	\end{proof}
	
	Combining (\ref{eq:L8}) and (\ref{eq:P2b}) establishes Theorem 2. In addition, with the notation (\ref{eq:14b}),
	we see that the explicit form of the $O(N^{-2/3})$ correction term is
	\begin{equation}\label{eq:rLaN}
		{1 \over N^{2/3}} \bigg (  - {1 \over 8}  {1 \over b^{2/3}} {\alpha^2 \over 1 + \alpha}
		\Big ( {1 \over 2} - {1 \over \beta} \Big )^2 {d^2 \over d x^2} \rho_{\infty,\beta,0}(x) +
		r_{\infty,\beta}^{(L,\alpha N)}(x) \bigg ).
	\end{equation}
	In the case $\beta = 2$, proceeding as in the Gaussian case of Remark \ref{R1}, we can check that this
	is consistent with $\rho_{\infty,\beta=2,1}^{(L,\alpha N)}$ as specified in (\ref{7.3}).
	
	\subsection*{The case $a $ fixed}
	We fix $a$ and make the replacement
	\begin{equation}\label{eq:a1}
		x \mapsto 1+ {x\over (2N)^{2/3}}.
	\end{equation}
	A simple calculation shows 
	\begin{equation}\label{eq:a2}
		(4Nx)^{a\beta/2} e^{-2\beta Nx}\,\bigg\rvert_{x \mapsto 1+x/(2N)^{2/3}} = (4N)^{a\beta/2}
		e^{-2\beta N - \beta(2N)^{1/3}x}
		\left(
			1+ {1\over N^{2/3}} {2^{1/3}\over 4} a\beta x + O(N^{-4/3})
			\right).
	\end{equation}
	Further, by use of Stirling's formula and the multiplication formula for the gamma function we deduce
	\begin{multline}\label{eq:a3}
		{W_{a+2,\beta,N-1}\over W_{a,\beta,N}} {1\over M_\beta(a+2/\beta-1,N-1,2/\beta) }
		\\
		= {1\over 2\pi} \left( {4\pi \over \beta} \right)^{\beta/2} N^{-2+\beta- a \beta /2} \Gamma(1+\beta/2)
		\prod_{j=2}^\beta {\Gamma(1+2/\beta) \over \Gamma(1+2j/\beta)}
		\Big (
		1 + O(N^{-1}) \Big ).
	\end{multline}
	
	With $a$ fixed we read of that the leading large $N$ term of (\ref{eq:L2}) is
		$$
		g(u)=4u-\log u +\log(1-u).
	$$
	 This has a double saddle point at $u_0=1/2$. Expanding around this point in the direction of the steepest descent by the change of variables $v =-2(2N)^{1/3}(u-u_0)$ then shows
	\begin{multline}\label{eq:a4}
		Nf(1/2-v/2(2N)^{1/3},1+x/(2N)^{2/3})
		= 2N+(2N)^{1/3}x+v^3/3-v x + (4-4/\beta-a)\log 2
		\\
		+ {av\over (2N)^{1/3}} + {2^{1/3}\over N^{2/3}} \left(  
		v^5/10 + (4-4/\beta-a)v^2/4 
		\right) + O(N^{-1}).
	\end{multline}
	
	By combining (\ref{eq:a2}), (\ref{eq:a3}), (\ref{eq:a4}) with (\ref{eq:L1}), for the Laguerre $\beta$-ensemble with $a$ fixed, the large $N$ expansion of the soft edge scaled density as defined in the variable (\ref{eq:a1}) is given by
	\begin{multline}\label{eq:L10a}
	2(2N)^{1/3} \rho_{N,\beta}^{(L,a)}\left(4N + 2(2N)^{1/3}x \right) 
		= \frac{1}{2\pi}\left(\frac{4\pi}{\beta} \right)^{\beta/2} \Gamma(1+\beta/2)
		\prod_{j=2}^\beta {\Gamma(1+2/\beta) \over \Gamma(1+2j/\beta)}
		\\
		\times \mathbb{K}_{\beta,\beta}
		\left[ 1+ \frac{1}{N^{1/3}}c_1^{(L,a)}(v;x,\beta)+\frac{1}{N^{2/3}}c_2^{(L,a)\ast}(v;x,\beta) +O(N^{-1})
		\right],
	\end{multline}
	where
	\begin{align}
		c_1^{(L,a)}(v;x,\beta) &=  {a\over 2^{1/3}}\sum_{j=1}^\beta v_j  \nonumber
		\\
		c_2^{(L,a)\ast}(v;x,\beta) &=  2^{1/3}\sum_{j=1}^\beta \left( v_j^5/10 + (1-1/\beta)v_j^2 \right) + {1\over 2} [c_1^{(L,a)}(v;x,\beta)]^2.
		\label{eq:L10b}
	\end{align}
	
	Immediate from (\ref{eq:L10b}) and (\ref{eq:K6x}) is that
	\begin{equation}\label{eq:LK}
		\mathbb{K}_{\beta,\beta}[c_1^{(L,a)}(v;x,\beta)](x)=- {a\over 2^{1/3}}\frac{d}{dx}K_{\beta,\beta}(x).
	\end{equation}
	Applying the translation $x\mapsto x+a/(2N)^{1/3}$ to (\ref{eq:L10a}) then shows that the scaled soft edge density in the variable listed second in 
	(\ref{eq:7.2}), up to order $O(N^{-2/3})$, is given by
	\begin{multline}\label{eq:L10}
	2(2N)^{1/3} \rho_{N,\beta}^{(L,a)}\left(4N + 2a + 2(2N)^{1/3}x \right) 
	= \frac{1}{2\pi}\left(\frac{4\pi}{\beta} \right)^{\beta/2} \Gamma(1+\beta/2)
	\prod_{j=2}^\beta {\Gamma(1+2/\beta) \over \Gamma(1+2j/\beta)}
	\\
	\times \mathbb{K}_{\beta,\beta}
	\left[ 1+ \frac{1}{N^{2/3}}c_2^{(L,a)}(v;x,\beta) +O(N^{-1})
	\right],
	\end{multline}
	where
	\begin{equation}
	c_2^{(L,a)}(v;x,\beta) =  2^{1/3}\sum_{j=1}^\beta \left( v_j^5/10 + (1-1/\beta)v_j^2 \right) . \label{eq:L11}
	\end{equation}
	This corroborates with Theorem 1 in the Laguerre case, and furthermore gives the leading correction as
	\begin{equation}\label{eq:rLa}
		{1 \over N^{2/3}} r_{\infty,\beta}^{(L,a)}(x), 
	\end{equation}
	where $r_{\infty,\beta}^{(L,a)}$ is specified according to (\ref{eq:14b}) with the substitution (\ref{eq:L11}). Note that this correction is independent of $a$.
	In the case $\beta = 2$ this latter property is already evident from the appropriate case of (\ref{7.3}). The procedure of Remark \ref{R1} can be
	used to reduce (\ref{eq:rLa}) with $\beta = 2$ to this functional form.

\section{Numerics}\label{S5}
         In this section we provide a method to determine the graphical form of the optimal $O(N^{-2/3})$ correction term for 
         the soft edge density of the Gaussian and Laguerre even $\beta$-ensembles. Multiple integral forms have been given in 
         (\ref{eq:14c}), (\ref{eq:rLaN}) and (\ref{eq:rLa}). However these are not well suited to accurate numerical evaluation.
         Instead we turn to known \cite{Fo93,FI10a,FR12} recursive properties relating to the average in (\ref{eq:5.2}), recalling that for
         $\beta$ even the latter is a polynomial of degree $\beta (N-1)$. As such we are making use of the broader theory of the
         Selberg integral; see \cite[Ch.~4]{Fo10}.
         
        For fixed parameters $(\lambda_1,\lambda_2,\lambda,\alpha)$, introduce the
         auxiliary function $I_p[w(t)](x) = I_p[w(t)](x;\alpha,\lambda)$ as the multiple integral
	\begin{multline}\label{eq:N1}
		I_p[w(t)](x) = {p!(N-p)!\over N!} \int_{\mathbb R} dt_1 \cdots \int_{\mathbb R} dt_N \, \prod_{l=1}^N w(t_l) \abs{x-t_l}^{\alpha-1} 
		\\
		\times \prod_{1\leq j<k\leq N} \abs{t_k-t_j}^{2\lambda} e_p(x-t_1,\ldots,x-t_N),
	\end{multline}
         where 
         \begin{equation*}
		e_p(t_1,\ldots,t_n) = \sum_{1\leq j_1<\ldots <j_p\leq n} t_{j_1} \cdots \, t_{j_p}
	\end{equation*}
	denotes the elementary symmetric polynomials. In the case $w(t) = t^{\lambda_1} (1 - t)^{\lambda_2} \chi_{0 < t < 1}$
	we know from \cite{Fo93} that these integrals satisfy the  differential-difference equation
	\begin{equation}\label{DDE}
		(N-p)E_p I_{p+1}(x) = (A_p x+B_p)I_p(x) - x(x-1) \frac{d}{dx}I_p(x) + D_p x(x-1)I_{p-1}(x),
	\end{equation}
	where
	\begin{align*}
		A_p&=(N-p)(\lambda_1+\lambda_2+2\lambda(N-p-1)+2\alpha)
		\\
		B_p&=(p-N)(\lambda_1+\alpha+\lambda(N-p-1))
		\\
		D_p&=p(\lambda(N-p)+\alpha)
		\\
		E_p&=\lambda_1+\lambda_2 + 1+ \lambda (2N-p-2)+\alpha.
	\end{align*}
	The immediate relevance of (\ref{eq:N1}) to the average (\ref{eq:5.2}) comes from the fact that
	\begin{equation*}
		I_N[ t^{\lambda_1} (1 - t)^{\lambda_2} \chi_{0 < t < 1}](x) \big\rvert_{\alpha=\beta,\lambda=\beta/2} \propto
		\left\langle
		\prod_{l=1}^N (x-x_1)^\beta
		\right\rangle_N^{(J)},
	\end{equation*}
	where the superscript $(J)$ indicates use of the Jacobi weight in (\ref{eq:5.2}).  In the case $\alpha = 1$, $p=0$, (\ref{eq:N1}) is
	independent of $x$, which provides an initial condition to iterate (\ref{DDE}) up to $p=N$. The crucial property
	\begin{equation*}
		I_N[w(t)](x) = I_0[w(t)](x)\bigg\rvert_{\alpha \to \alpha +1}
	\end{equation*}
	then gives the value of the polynomial in the case $\alpha = 2$, $p=0$, and the iteration can be continued.
	
	Introduce the notation
	$$
	G_p(x) = I_p[e^{-\beta t^2/2}](x), \qquad L_p(x) = I_p[t^{\beta a/2} e^{- \beta t/2} \chi_{t>0}](x)
	$$
	for the Gaussian and Laguerre cases of (\ref{eq:N1}). A simple change of variables shows 
	\begin{align}
		G_p(x) &= (-1)^p    \lim_{n\to\infty} 2^{\lambda_1 N} (2n)^{N\alpha + \lambda N(N-1) + p} I_p\left({1 \over 2} -{x\over 2n}\right)\bigg\rvert_{\lambda_1=\lambda_2= \lambda n^2}
		\\
		L_p(x) &= \lim_{n\to\infty} n^{N(\lambda_1 + \alpha) + \lambda N(N-1) + p} I_p(x/n)\bigg\rvert_{\lambda_1=\lambda a,\lambda_2= \lambda n}.
	\end{align}
	Applying this limiting procedure to (\ref{DDE}) gives
	\begin{equation}\label{eq:R1}
		\lambda (N-p)G_{p+1}(x) = \lambda (N-p)x G_p(x) + {1\over 2}\frac{d}{dx}G_p(x) - {p(\lambda(N-p)+\alpha)\over 2}G_{p-1}(x)
	\end{equation}
	and
	\begin{equation}\label{eq:R2}
		\lambda (N-p)L_{p+1}(x) = (\lambda (N-p)x+B_p)L_p(x) + x\frac{d}{dx}L_p(x) -D_p xL_{p-1}(x),
	\end{equation}
	where $\lambda = \beta/2$ and $\lambda_1 = a \beta/ 2$.
The constant (i.e.~$x$ independent) values for $\alpha = 1$, $p=0$ can be read off from the normalisations appearing in
(\ref{eq:12.0}) and (\ref{Bv}), which being limiting cases of the Selberg integral can be expressed in terms of products of
gamma functions. Thus $G_0 \big\rvert_{\alpha=1} = G_{2 \lambda,N}$ and $L_0 \big\rvert_{\alpha=1} = W_{a,2 \lambda, N}$.

	Figures \ref{Fig1}, \ref{Fig1a}, \ref{Fig1b} and \ref{Fig2} were generated by substituting the appropriate soft edge variables
	\begin{equation}\label{eq:S1}
		s_{x,\beta}=
		\begin{cases}
		\sqrt{2N} + \frac{1}{\sqrt{2N}}\left(\frac{1}{2}-\frac{1}{\beta}  \right)+{x\over \sqrt{2}N^{1/6}},\quad &\text{Gaussian}
		\\
		4N+2a+2(2N)^{1/3}x,\quad &\text{Laguerre}
		\\
		N(\sqrt{1+\alpha}+1)^2 + {\alpha\over \sqrt{1+\alpha}} \left(\frac{1}{2}-\frac{1}{\beta}  \right) + 2(bN)^{1/3}x,
		 \quad &\text{Laguerre } a=\alpha N
		\end{cases}
	\end{equation}
	into (\ref{eq:R1}) and (\ref{eq:R2}). With $s_{x,\beta}$, $s_x'$ defined by (\ref{eq:S1}), (\ref{eq:7.2}) (with the exception for the Laguerre case with $a$ fixed, where we take instead $s_x'=4N+2(2N)^{1/3}x$), Figures \ref{Fig1}, \ref{Fig1a} and \ref{Fig1b} plot
	\begin{equation}\label{eq:F1}
		\left(
		{1\over N_2^{2/3}} - {1\over N_1^{2/3}}
		\right)^{-1}
		\left( 
		{\partial s_x \over \partial x}\bigg\rvert_{N=N_1} \rho_{N_1,\beta}^{\#}(s_{x,\beta}) 
		-
		{\partial s_x \over \partial x}\bigg\rvert_{N=N_2} \rho_{N_2,\beta}^{\#}(s_{x,\beta}) \right).
	\end{equation}
	For some large $N_1,N_2 \in \mathbb{N}$ values, the successive differences given by (\ref{eq:F1}) admit, as seen graphically, the same functional form and are $O(1)$. 
	(\ref{eq:F1}) plots the corrections (\ref{eq:14c}), (\ref{eq:rLa}), (\ref{eq:rLaN}) up to some error, to leading order, proportional to
	$
	\left(
	{1\over N_2^{2/3}} - {1\over N_1^{2/3}}
	\right)^{-1}
	\left(
	{1\over N_2} - {1\over N_1}
	\right).
	$
	Figure \ref{Fig2} plots
	\begin{equation}\label{eq:F2}
		{N^{1/3} \over k^\#(\beta)} 
		{\partial s_x\over \partial x}
		\left( \rho_{N,\beta}^{\#}(s_{x}') - \rho_{N,\beta}^{\#}(s_{x,\beta}) \right)
	\end{equation}
	where
	$$
		k^\#(\beta) = 
		\begin{cases}
		\frac{1}{2}-\frac{1}{\beta} , \quad &\text{Gaussian}
		\\
		{a\over 2^{1/3}}, \quad &\text{Laguerre}
		\\
		{1\over 2b^{1/3}} {\alpha\over \sqrt{1+\alpha}} \left(\frac{1}{2}-\frac{1}{\beta} \right) \quad &\text{Laguerre } a=\alpha N
		\end{cases}
	$$
	and shows the approximate functional form of the correctional term of order $N^{-1/3}$ due to the scaling (\ref{eq:7.2}) which bears resemblance to the derivative of (\ref{eq:5.2}). This further emphasises the results found in Proposition \ref{P1}, \ref{P2} and (\ref{eq:LK}). The small displacement relative to the derivative in the displayed examples -- which we take as a finite N artifact –- is not a feature of the previous figures (right captions) due to the fundamentally different nature of what is being plotted.

	\begin{figure}[H]
	\centering
	\includegraphics[width=7cm]{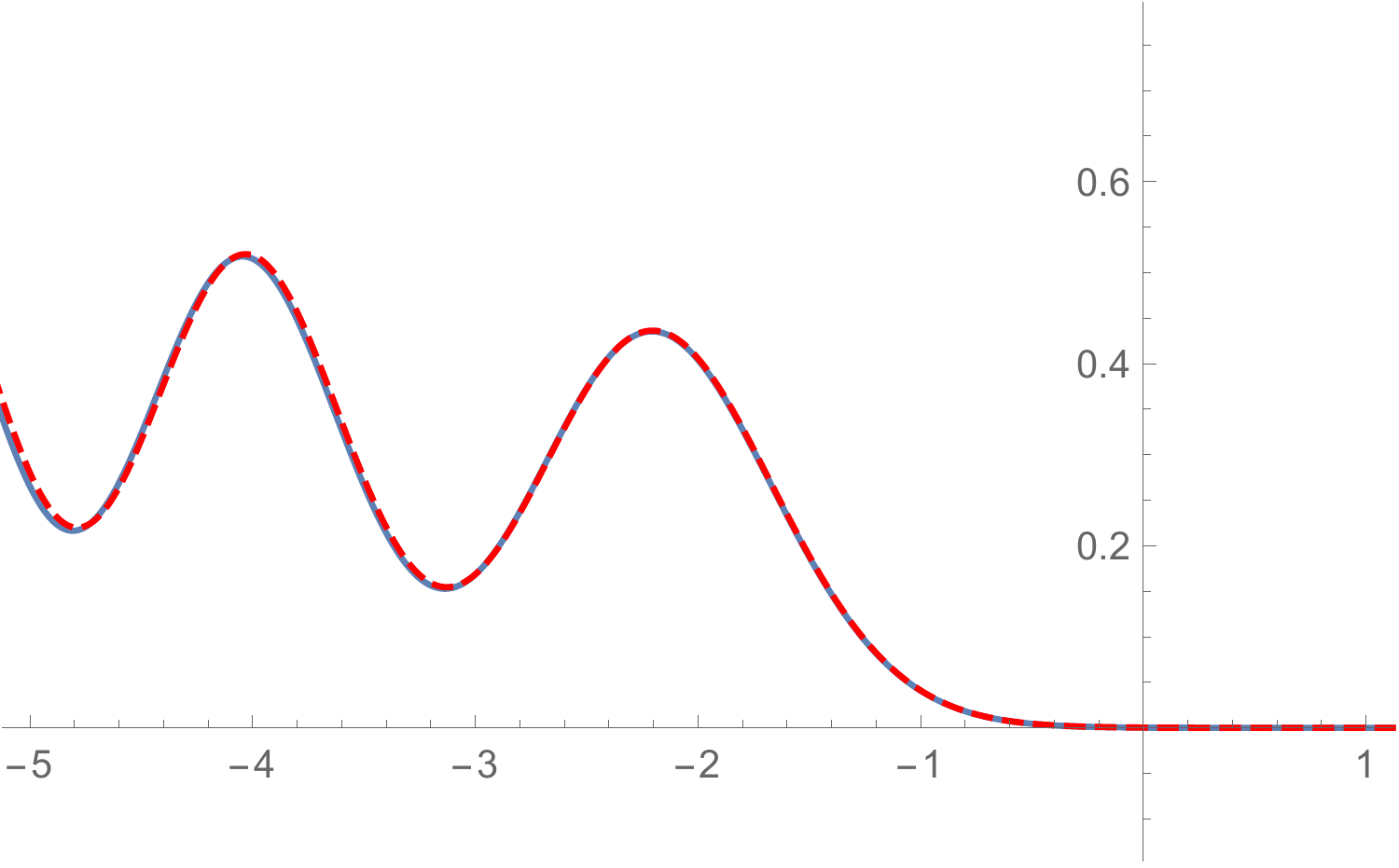}
	\includegraphics[width=7cm]{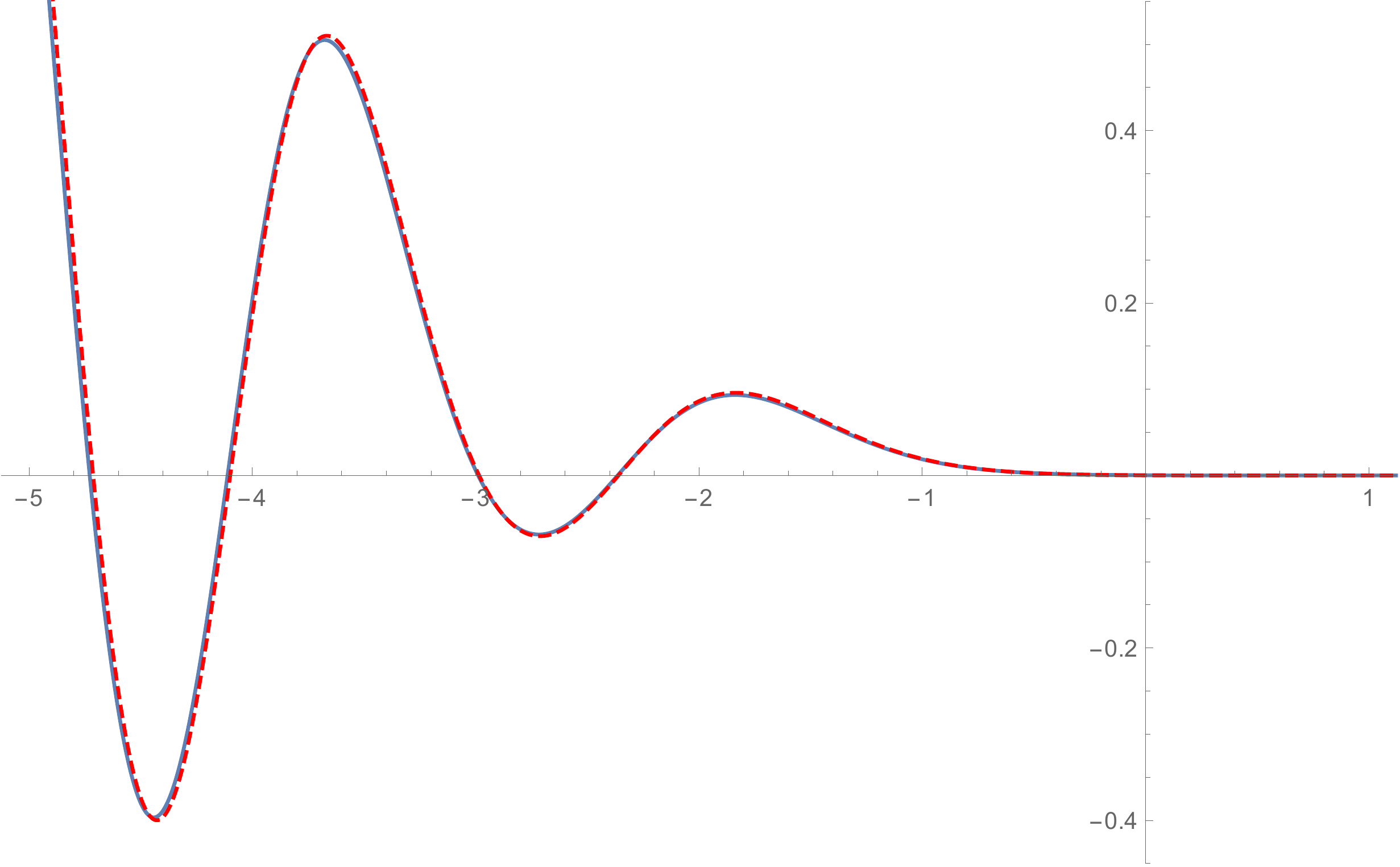}
	\caption{Left: $1/(\sqrt{2} N^{1/6})\rho_{N,\beta}^{(G)}(s_{x,\beta})$ with values $\beta=6$, $N=30$ (blue solid line) and $40$ (red dashed line).
		Right: the difference (\ref{eq:F1}) with values $\beta=6$, $(N_1,N_2)=(30,40)$ (blue solid line) and $(N_1,N_2)=(40,50)$ (red dashed line).
	}
	\label{Fig1}
	\end{figure}

	\begin{figure}[H]
	\centering
	\includegraphics[width=7cm]{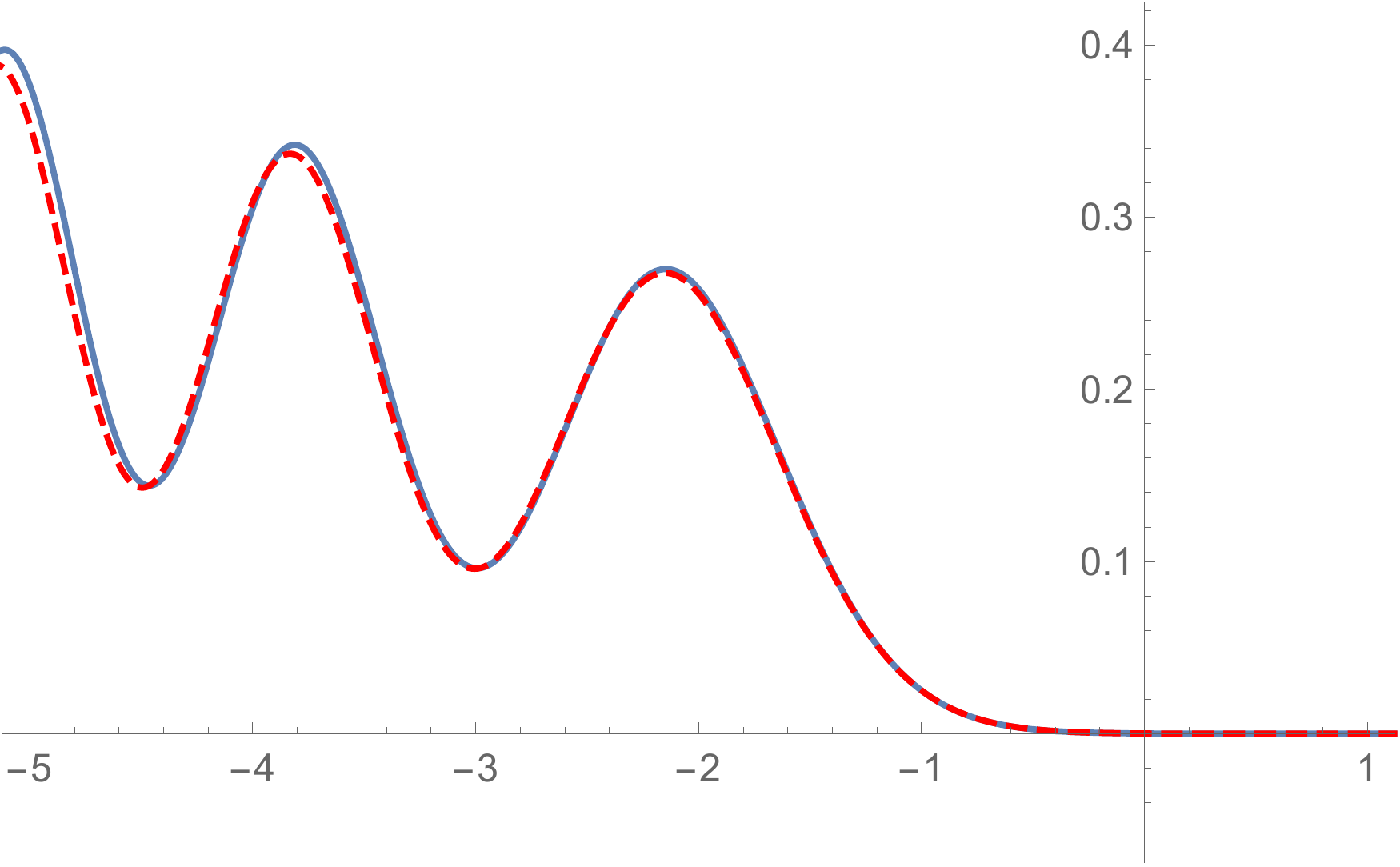}
	\includegraphics[width=7cm]{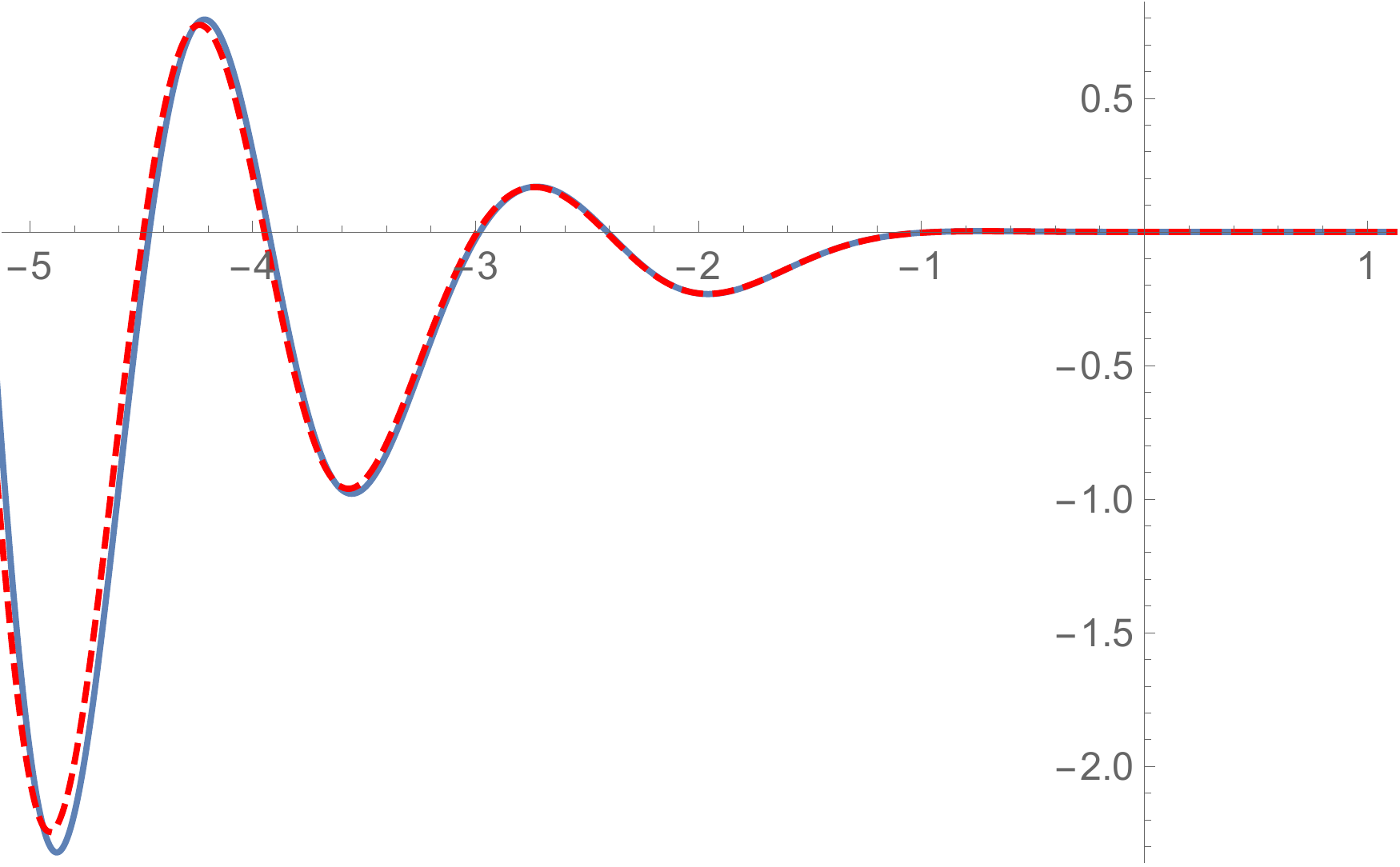}
	\caption{Left: $2(2N)^{1/3} \rho_{N,\beta}^{(L,a)}(s_{x,\beta})$ with values $\beta=6$, $a=0.5$, $N=40$ (blue solid line) and $50$ (red dashed line).
		Right: the difference (\ref{eq:F1}) with values $\beta=6$, $a=0.5$, $(N_1,N_2)=(40,50)$ (blue solid line) and $(N_1,N_2)=(50,60)$ (red dashed line).
	}
	\label{Fig1a}
	\end{figure}

	\begin{figure}[H]
	\centering
	\includegraphics[width=7cm]{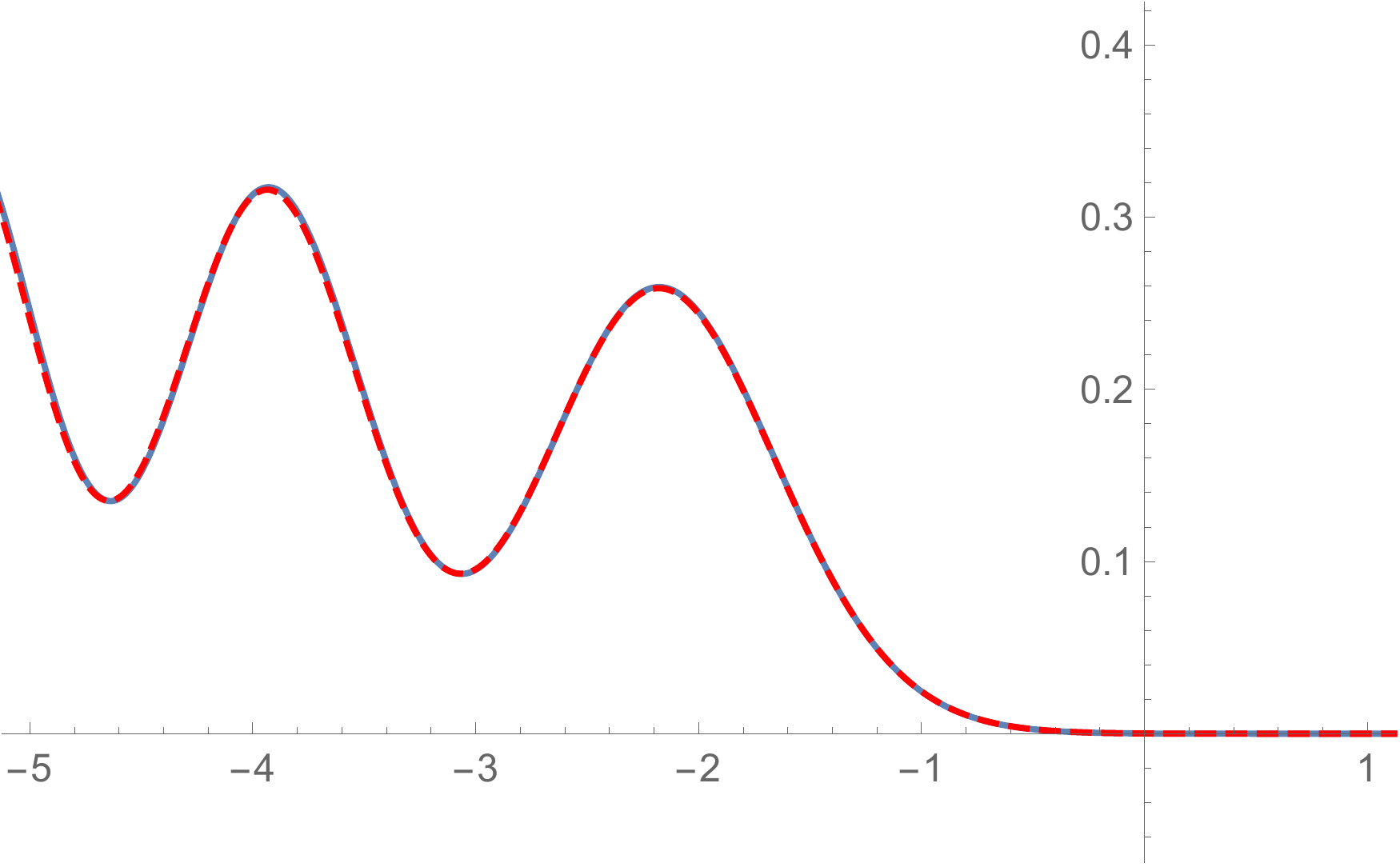}
	\includegraphics[width=7cm]{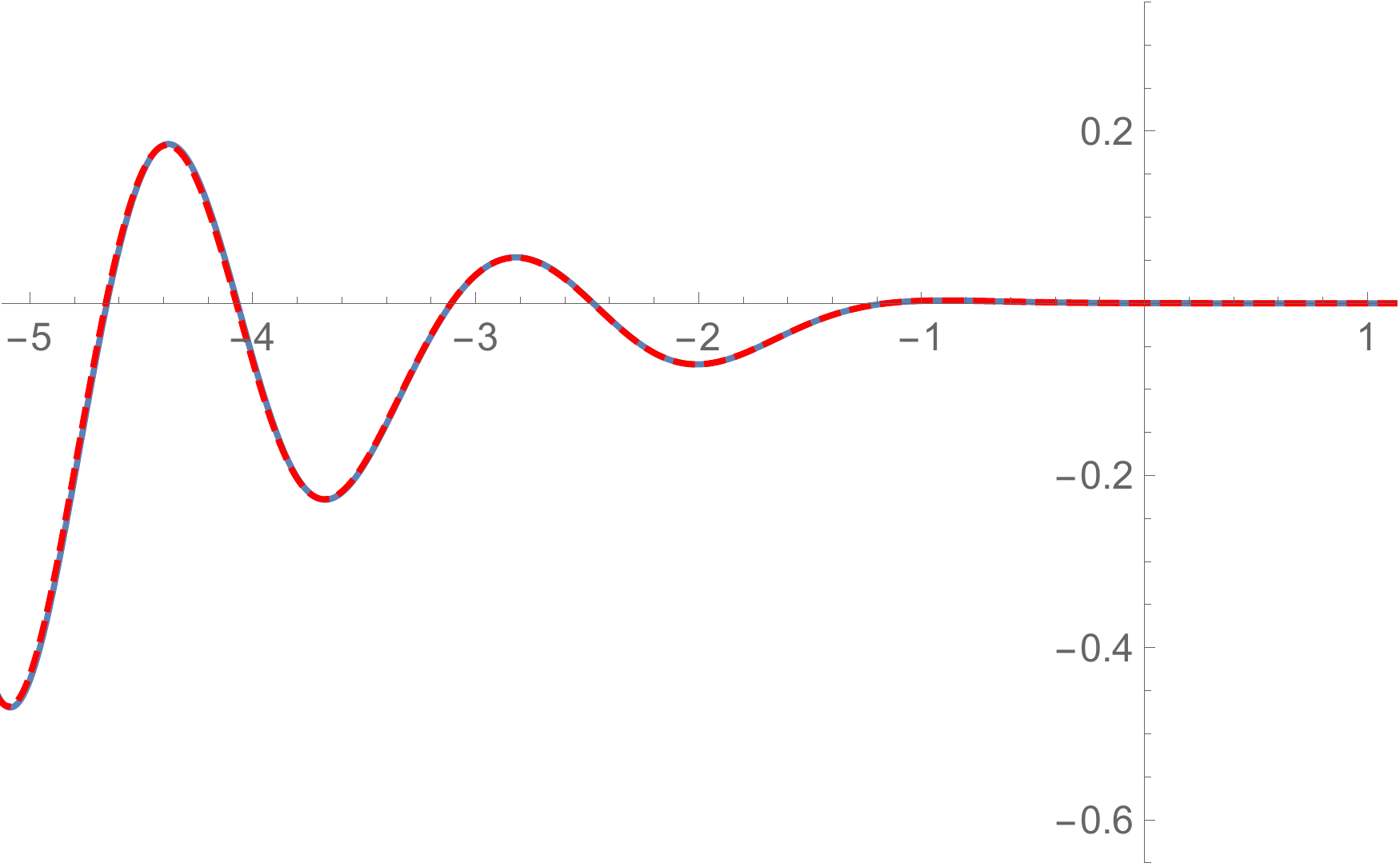}
	\caption{Left: $2(bN)^{1/3} \rho_{N,\beta}^{(L,\alpha N)}(s_{x,\beta})$ with values $\beta=6$, $\alpha=10$, $N=40$ (blue solid line) and $50$ (red dashed line).
		Right: the difference (\ref{eq:F1}) with values $\beta=6$, $\alpha=10$, $(N_1,N_2)=(40,50)$ (blue solid line) and $(N_1,N_2)=(50,60)$ (red dashed line).
	}
	\label{Fig1b}
\end{figure}

	\begin{figure}[H]
	\centering
	\includegraphics[width=7cm]{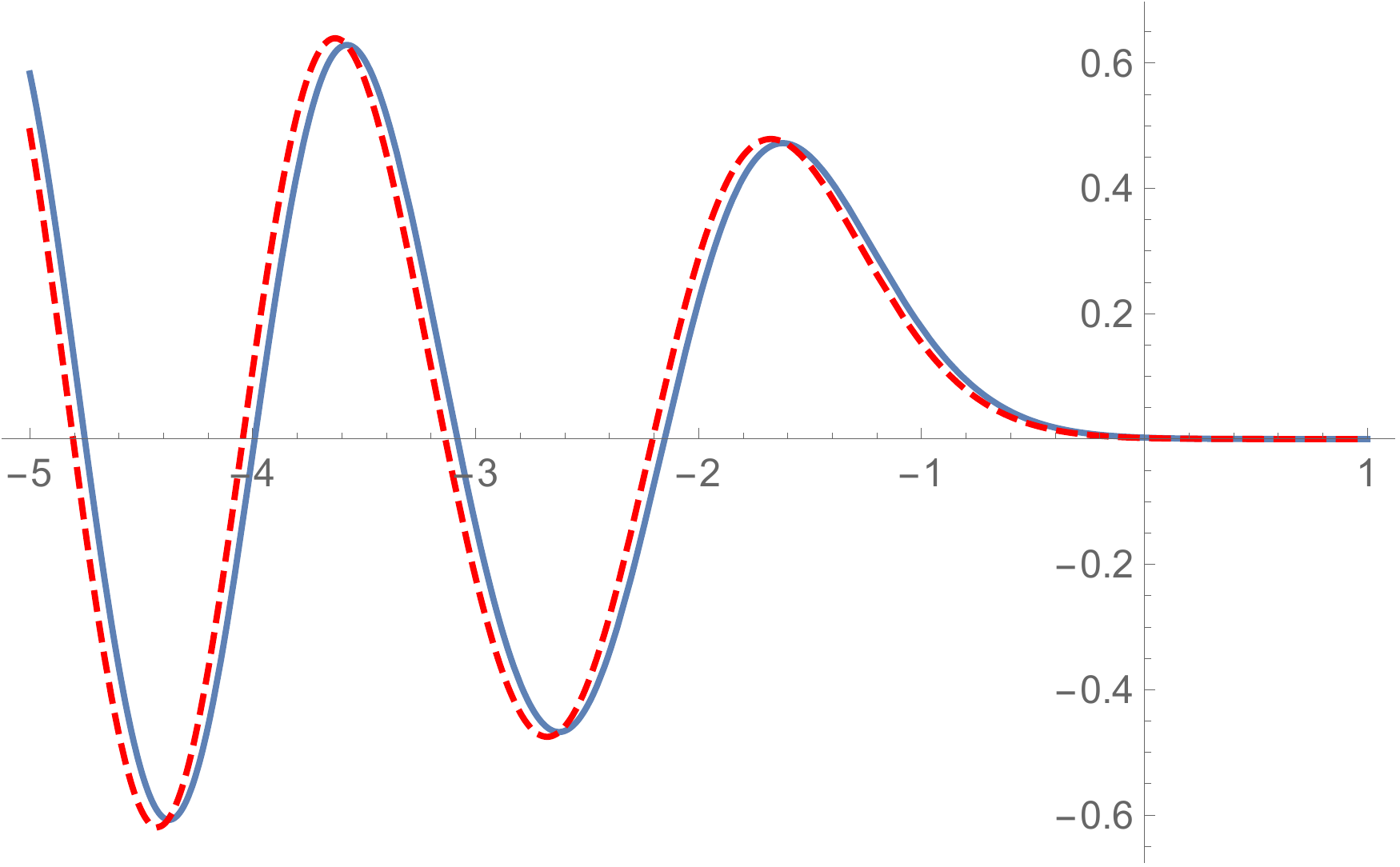}
	\includegraphics[width=7cm]{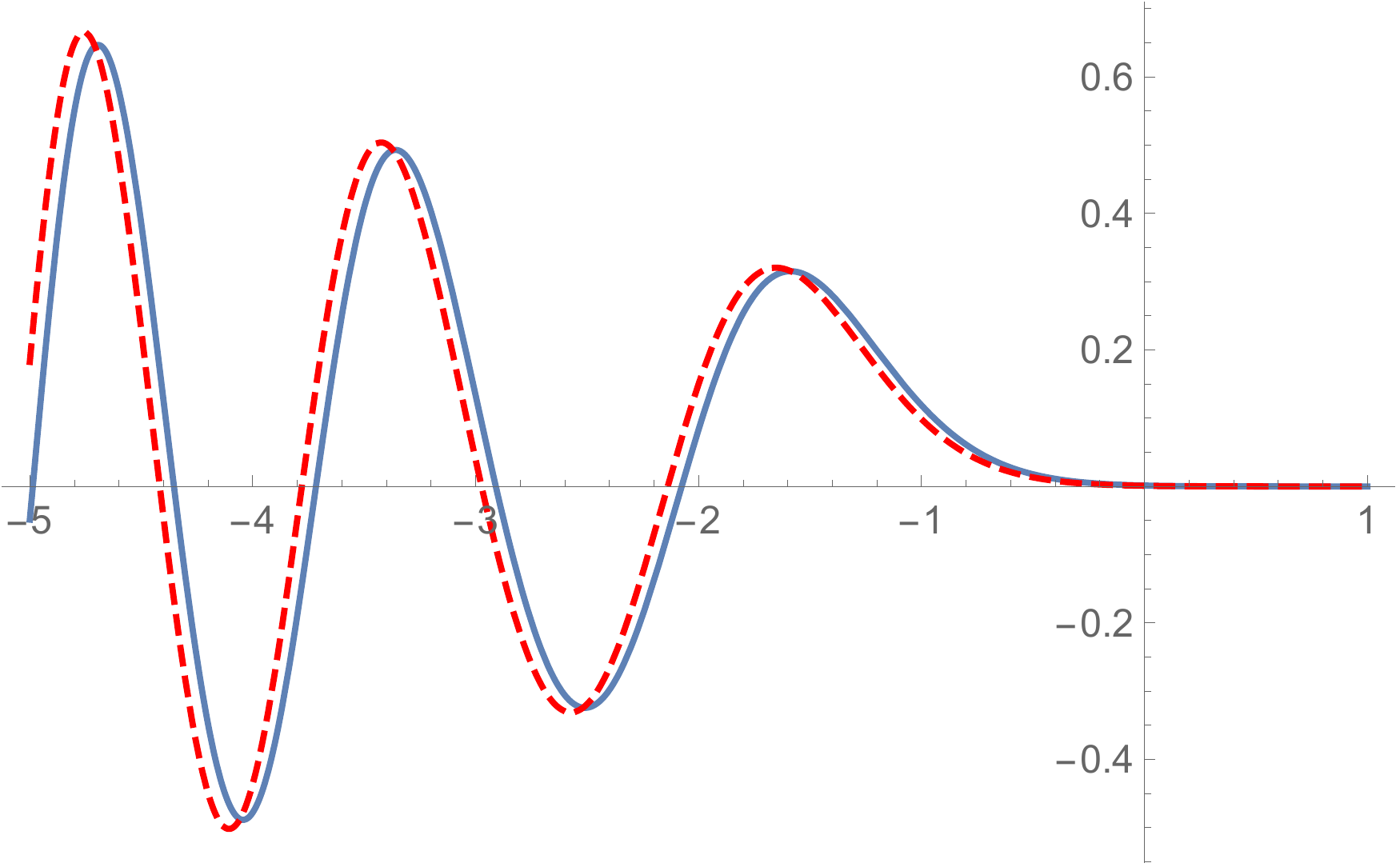}
	\includegraphics[width=7cm]{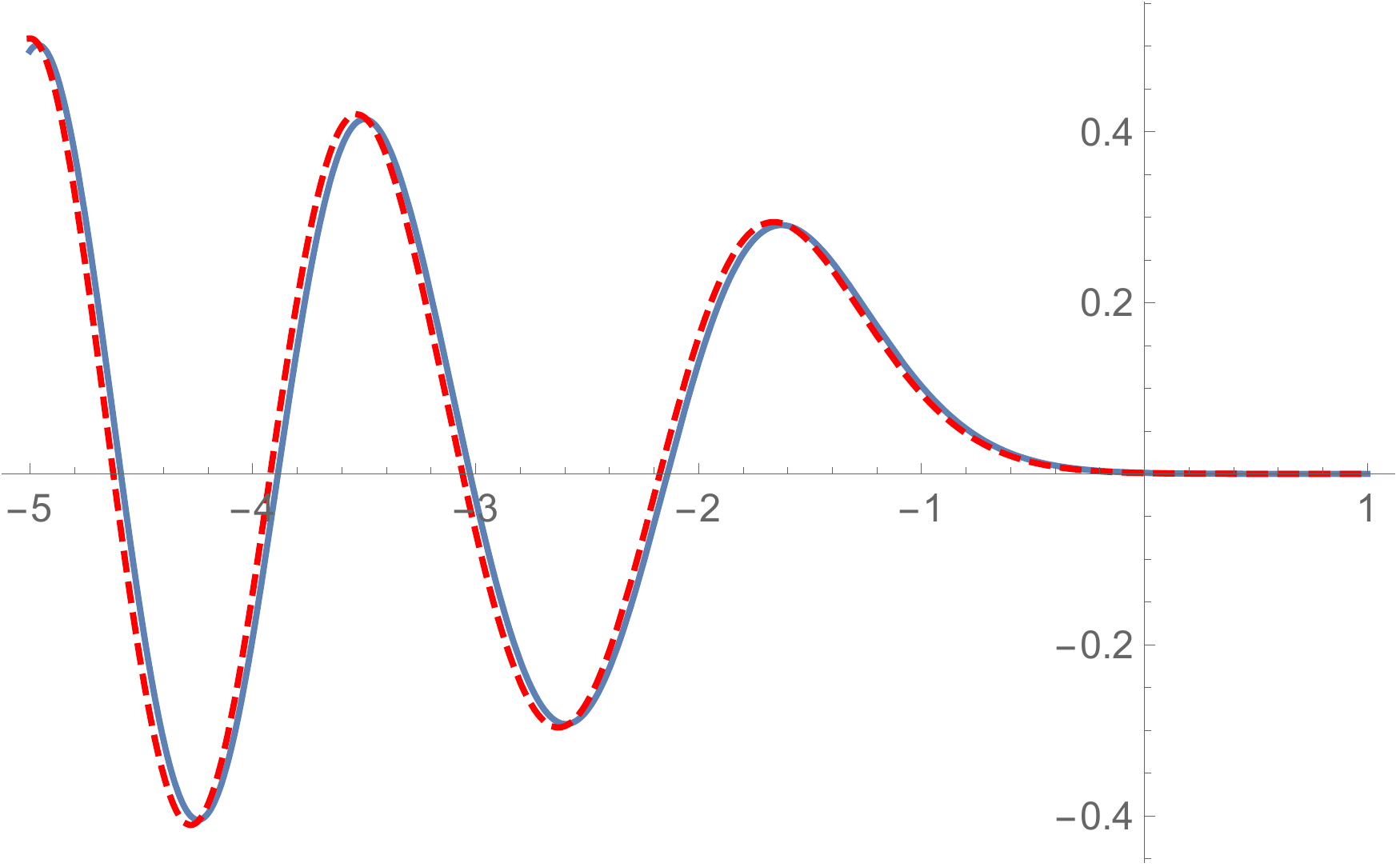}
	\caption{The $O(N^{-1/3})$ correction (\ref{eq:F2}) (blue solid line) and $-{\partial s_x \over \partial x} \frac{d}{dx}\rho_{N,\beta}^{\#}(s_{x,\beta})$ (red dashed line) with values $\beta=6$ and $N=30$. First Row (from left to right): Gaussian case, Laguerre case with $a=0.5$ fixed. Second Row: Laguerre case with $a=\alpha N$ and $\alpha=10$.
	}
	\label{Fig2}
	\end{figure}
	
	\section{More general invariant ensembles}
	In our previous study \cite{FT18}, the question of the optimal soft edge scaling for complex Hermitian Wigner matrices was briefly addressed.
	We recall that an Hermitian   random matrix is termed a Wigner ensemble if all its diagonal entries (which must be real) are independently
	chosen from the same zero mean, finite variance distribution, and all its upper triangular entries (which may be complex)
	chosen similarly. The Gaussian unitary ensemble, corresponding to the Gaussian case of \eqref{eq:4.1} with $\beta = 2$, is both a unitary
	invariant ensemble, and a complex Wigner matrix: matrices from this ensemble can be generated by choosing the diagonal entries from
	the normal distribution N$[0,1/\sqrt{2}]$ and the upper triangular entries from N$[0,1/2] + i {\rm N}[0,1/2]$. In \cite{FT18} numerical simulations were     performed on the particular complex Wigner matrices $Y = {1 \over 2} (X + X^\dagger)$, where all entries of $X$ are chosen independently and
	uniformly from the set of four values ${1 \over \sqrt{2}}(\pm 1 \pm i)$. The off diagonal entries of $Y$ then have mean zero and variance one half,
	as for the GUE, guaranteeing (see e.g.~\cite{PS11}) that to leading order the largest eigenvalue occurs at $\sqrt{2N}$. By the use of simulation,
	evidence was presented that a scaled and centred variable can be identified such that the leading correction to the limiting distribution
	of the soft edge scaled largest eigenvalue is $O(N^{-2/3})$. Of course the natural question is to ask if this effect holds more generally.
	
	What then for the case of general invariant ensembles beyond the classical cases? In the case $\beta = 2$, and for the family of weights $w_2(x) =
	e^{-N x^{2m}}$, as noted in \cite[Eq.~(1.9)]{KSSV14}, results of \cite{DG07} give the asymptotic formula for the soft edge scaled
	correlation kernel (\ref{eq:7.1})
	\begin{equation}
	{ 1 \over N^{2/3} \gamma} K \Big ( b + {s \over N^{2/3} \gamma}, 	 b + {t \over N^{2/3} \gamma} \Big ) =
	 { \Ai(s)\Ai'(t)-\Ai'(s)\Ai(t)\over s-t } + O(N^{-2/3})O(e^{-C(s+t)})
	 \end{equation}
	 where $b,\gamma$ depend on $m$. Thus once again there is a choice of scaling variables such that the correction to the correlations at the soft
	 edge is $O(N^{-2/3})$.

\section*{Acknowledgements}
This work is part of a research program supported by the Australian Research Council (ARC) through the ARC Centre of Excellence for Mathematical and Statistical frontiers (ACEMS). PJF also acknowledges partial support from ARC grant DP170102028, and AKT acknowledges the support of a Melbourne postgraduate award.


\providecommand{\bysame}{\leavevmode\hbox to3em{\hrulefill}\thinspace}
\providecommand{\MR}{\relax\ifhmode\unskip\space\fi MR }
\providecommand{\MRhref}[2]{%
  \href{http://www.ams.org/mathscinet-getitem?mr=#1}{#2}
}
\providecommand{\href}[2]{#2}

\end{document}